\documentclass[10pt,a4paper]{article}
\usepackage{authblk}
\usepackage{geometry}
\usepackage{amsmath,amssymb,amsthm}
\usepackage[T1]{fontenc}
\usepackage{textgreek}
\usepackage{hyperref}

\usepackage[linesnumbered,ruled,vlined,noresetcount]{algorithm2e}
\DontPrintSemicolon
\SetKwProg{Fn}{function}{:}{end}
\SetKwProg{when}{when }{ \textbf{do}}{done}
\SetKwProg{For}{for }{ \textbf{do}}{done}
\let\oldnl\nl% Store \nl in \oldnl
\newcommand{\nonl}{\renewcommand{\nl}{\let\nl\oldnl}}% Remove line number for one line

\usepackage{graphicx}
\usepackage{hyperref}

\newtheorem{lemma}{Lemma}
\newtheorem{theorem}{Theorem}

\theoremstyle{definition}

 \newtheorem{definition}{Definition}

\newcommand{\ie}{\textit{i.e., }}
\newcommand{\eg}{\textit{e.g., }}

\newcommand{\CIW}{\mathtt{CIW}}       
\newcommand{\CIG}{\mathtt{CIG}}
\newcommand{\CIK}{\CIW_{n,k}}

\newcommand{\etal}{et~al.}
\newcommand{\hit}{\mathcal{H}}

\newcommand{\yes}{\mathtt{yes}}
\newcommand{\no}{\mathtt{no}}
\newcommand{\cnt}{\mathtt{cnt}}
\newcommand{\leader}{\mathtt{leader}}
\newcommand{\mode}{\mathtt{mode}}
\newcommand{\phase}{\mathtt{phase}}
\newcommand{\vl}{V_{L}}
\newcommand{\group}{\mathtt{group}}
\newcommand{\sz}{\mathtt{size}}
\newcommand{\token}{\mathtt{token}}
\newcommand{\RReset}{\mathtt{Reset}}

\title{Complete Graph Identification in Population Protocols\thanks{
This work is supported by JST FOREST Program JPMJFR226U.
}} 
\date{}

\author[1]{Haruki Kanaya}
\affil[1]{Nara Institute of Science and Technology, Nara, Japan}
\author[2]{Yuichi Sudo\thanks{Corresponding Author: sudo@hosei.ac.jp}}
\affil[2]{Hosei University, Tokyo, Japan}

\begin{document}

\maketitle

%TODO mandatory: add short abstract of the document
\begin{abstract}
We consider the population protocol model where indistinguishable state machines, referred to as agents, communicate in pairs. The communication graph specifies potential interactions (\ie communication) between agent pairs. This paper addresses the complete graph identification problem, requiring agents to determine if their communication graph is a clique or not. We evaluate various settings based on: (i) the fairness preserved by the adversarial scheduler -- either global fairness or weak fairness, and (ii) the knowledge provided to agents beforehand -- either the exact population size $n$, a common upper bound $P$ on $n$, or no prior information. Positively, we show that $O(n^2)$ states per agent suffice to solve the complete graph identification problem under global fairness without prior knowledge. With prior knowledge of $n$, agents can solve the problem using only $O(n)$ states under weak fairness. Negatively, we prove that complete graph identification remains unsolvable under weak fairness when only a common upper bound $P$ on the population size $n$ is known.
\end{abstract}

\section{Introduction}
\label{sec:intro}
The population protocol model, proposed by Angluin et al.~\cite{AngluinADFP2006}, is a widely recognized computational model in distributed computing. Originally introduced to represent passively mobile sensor networks, this model is also applicable to chemical reactions systems, molecular computing, and similar systems. In this model, $n$ state machines, called \emph{agents}, form a population. These agents are indistinguishable, \ie they lack unique identifiers. The states of agents are updated through pairwise communications, called \emph{interactions}. The \emph{communication graph} specifies which agent pairs may interact. 
Specifically, the communication graph $G=(V,E)$ is a simple and weakly-connected digraph. 
Each arc $(u,v) \in E$ implies that $u$ and $v$ can interact where $u$ and $v$ serve as the initiator and responder, respectively. 
An adversarial scheduler selects exactly one interactable ordered pair of agents at each time step.
When two agents $u$ and $v$ are selected by the scheduler, they update their states according to their states and their roles.
The two roles, initiator and responder, may be helpful to break the symmetry: two agents with the same state may get different states after they interact. 
Throughout this paper, we denote the number of agents by $n=|V|$ and the set of non-negative integers by $\mathbb{N}$.

The adversarial scheduler must preserve some constraint, either \emph{global fairness} or \emph{weak fairness}.
Global fairness, the most commonly assumed fairness in population protocol literature, ensures that if a configuration $C$ appears infinitely often in an execution, every configuration reachable from $C$ must also appear infinitely often. As we will discuss later, a protocol $\mathcal{A}$ solves a problem $\mathcal{P}$ under global fairness if and only if $A$ solves $\mathcal{P}$ with probability $1$ under the \emph{uniformly random scheduler}, which selects any interactable ordered pair of agents uniformly at random at each time step. Thus, in terms of solvability, the globally fair scheduler can be considered equivalent to the uniformly random scheduler. Weak fairness is straightforward; it merely ensures that an interaction between initiator $u$ and responder $v$ occurs infinitely often if $(u,v) \in E$.

The graph class identification problem was first studied by Angluin, Aspnes, Chan, Fischer, Jiang, and Peralta~\cite{Angluin.10.1007/11502593_8}.
This problem requires the agents to determine whether the communication graph belongs to a given graph class.
The importance of this problem stems from the fact that many protocols in the literature are designed for specific graphs. 
Notably, most studies on population protocols assume a complete communication graph.
Angluin \etal~\cite{Angluin.10.1007/11502593_8} showed that under the global fairness, the population can identify various graph classes, \ie (i) directed lines, (ii) directed rings, (iii) directed stars, (iv) directed trees, (v) graphs with both in-degree and out-degree of every node bounded by a given integer $k$, (vi) graphs containing a given subgraph, (vii) graphs containing any directed cycle, and (viii) graphs containing any odd-length directed cycle.

Yasumi, Ooshita, and Inoue~\cite{yasumi_et_al:LIPIcs.OPODIS.2021.13} studied graph class identification, assuming the communication graph $G=(V,E)$ is \emph{undirected}, \ie $(u,v) \in E$ if and only if $(v,u) \in E$ for all $u,v \in V$. They showed that the algorithms given by Angluin \etal~\cite{Angluin.10.1007/11502593_8} can be extended to identify (undirected) lines, rings, stars, and bipartite graphs under global fairness, by treating two arcs $(u,v)$ and $(v,u)$ as a single undirected edge $\{u,v\}$.
Moreover, they gave two algorithms that identify trees and $k$-regular graphs for a given $k$, respectively,
under global fairness.
The latter
requires all agents to know a common upper bound $P$ on the population size $n=|V|$.
They also explored the solvability under weak fairness, revealing that: (i) lines, rings, and bipartite graphs are unidentifiable even if all agents know the exact number of agents $n$; (ii) stars can be identified if all agents know the exact $n$, but not under the knowledge of a common upper bound $P$ on $n$. The unidentifiability of rings implies that $k$-regular graphs cannot be identified under weak fairness when $k=2$, even with the knowledge of exact $n$. The identifiability of $k$-regular graphs under weak fairness for $k \neq 2$, including when $k=n-1$ (\ie complete graphs), remains an open question.

\begin{table}[t]
    \caption{Complete Graph Identification Protocols. 
    Time complexities are evaluated as the number of steps until the output of all agents stabilize under the uniformly random scheduler, bounded both in expectation and with high probability.
%The number of agents is denoted by $n$.
A given upper bound on $n$ is denoted by $P$, while $k$ is a design parameter satisfying $k\in[1,n]$.
See Section \ref{sec:contribution} for the definition of $\hat{G}$ and $\hit_{\hat{G}}$.}
    \label{table:t1}
    \centering
    \begin{tabular}{c c c c c }
        \hline
           & fairness &  knowledge & states & time complexity\\
        \hline                
        \cite{yasumi_et_al:LIPIcs.OPODIS.2021.13} &global & $n$ & $O(n\log{n})$ & exponential \\$\CIG$   & global & none & $O(n^2)$ & $O((\hit_{\hat{G}}+n^2)n\log{n})$\\        
        $\CIW_n$   & weak & $n$ & $O(n)$ & $O(n^3 \log n)$\\        
        $\CIW_{n,k}$   & weak & $n$ & $O(nk2^k)$ & $O\left(\frac{n^3}{k}\log n \right)$\\        
        \hline 
        impossible  & weak & $P$ & - & - \\        
        \hline 
    \end{tabular}
    
\end{table}

\subsection{Our Contributions}
\label{sec:contribution}
This paper focuses on identifying complete graphs, \ie determining whether the communication graph forms a clique.
As previously mentioned, most studies on population protocols assume that the communication graph is a complete graph. Thus, the identifiability of complete graphs is of significant importance.

We clarify the feasibility of identifying complete graphs based on fairness and prior knowledge of the population size. The summary of our results is listed in Table \ref{table:t1}.
First, we show that under global fairness, complete graphs can be identified without any prior knowledge of the population size. Specifically, we present an algorithm $\CIG$, which identifies complete graphs using $O(n^2)$ states per agent. 
Next, we show that even under weak fairness, complete graphs can be identified with exact knowledge of $n$: we introduce an algorithm $\CIW_n$ that identifies complete graphs with $O(n)$ states per agent under weak fairness.
The requirement for exact knowledge of the population size is justified by the fact that complete graphs cannot be identified under weak fairness even if the agents are aware of a common upper bound $P$ on $n$.
This impossibility directly follows from a lemma provided by Yasumi \etal~\cite{yasumi_et_al:LIPIcs.OPODIS.2021.13} (Lemma 5 in \cite{yasumi_et_al:LIPIcs.OPODIS.2021.13}).

In addition to the number of states, we also evaluate the time complexity of our protocols. In line with the conventions of population protocols, we measure time complexity as the number of interactions required until all agents stabilize their outputs under the uniformly random scheduler. The time complexities of our two protocols, $\CIG$ and $\CIW_n$, are $O((\hit_{\hat{G}}+n^2)n \log n)$ and $O(n^3 \log n)$ interactions both in expectation and with high probability, respectively. 
Here, 
$\hat{G}=(V,\hat{E})$ is the (non-simple) undirected graph obtained by
replacing each arc $(u,v) \in E$ with an undirected edge $\{u,v\}$,
and $\hit_{G'}$ denotes the maximum hitting time of an undirected graph $G'$, \ie the maximum expected number of moves required for a token performing the simple random walk on $G'$ to move from any node $s$ to any node $t$, where the maximum is taken over all possible pairs of $s$ and $t$.
Note that $\hat{E}$
is a multiset where $\{u,v\}$ appears at most twice for any pair $u,v \in V$.
Since the hitting time of any connected undirected graph $G'=(V',E')$ is known to be $O(|E'|\cdot D_{G'})$, where $D_{G'}$ is the diameter of $G'$, the time complexity of $\CIG$ is $O(n^4 \log n)$ interactions both in expectation and with high probability. With the precise knowledge of $n$, $\CIW_n$ achieves lower time complexity (\ie $O(n^3 \log n)$ interactions) and fewer states per agent (\ie $O(n)$ states). 
By introducing a design parameter $k \in [1,n]$, we show that $\CIW_n$ can be generalized to a protocol $\CIW_{n,k}$, which identifies complete graphs within $O((n^3 \log n)/k)$ interactions both in expectation and with high probability and requires $O(nk2^k)$ states per agent.

If we assume exact knowledge of the population size $n$, we can utilize Yasumi \etal's protocol for identifying $k$-regular graphs
\cite{yasumi_et_al:LIPIcs.OPODIS.2021.13} to identify complete graphs by setting $k = n-1$. However, our protocol $\CIW_n$ offers several advantages for identifying complete graphs:
(i) $\CIW_n$ requires only weak fairness, whereas Yasumi \etal's protocol requires global fairness,
(ii) $\CIW_n$ uses fewer agent states,
(iii) $\CIW_n$ is faster.
(See Table \ref{table:t1}.)
Although Yasumi \etal~\cite{yasumi_et_al:LIPIcs.OPODIS.2021.13} focus only on solvability and the number of states, and do not address time complexity, their protocol requires exponential time when $k = \Omega(n)$. Specifically, it does not stabilize before every agent $u$ experiences an event in which its $k$ most recent interactions involve exactly $k$ different neighbors.
If the degree of $u$ is exactly $k$, this event requires at least $\Omega(k^k/k!) = \Omega(e^k/k^{1/2})$ interactions in expectation due to Stirling's approximation $k! > \sqrt{2\pi k}(k/e)^k$. Therefore, for our purpose (\ie $k = n-1$), this protocol requires
$\Omega(e^n/\sqrt{n})$ expected interactions.
It may be worth mentioning that our protocol, designed for \emph{directed communication graphs}, also functions under the \emph{undirected} communication graphs assumed by Yasumi \etal{} By definition, undirected communication graphs are simply special cases of directed graphs.

In the original model of population protocols, agents are defined as \emph{state machines}, and the set of agent states must be specified in advance when designing a protocol. Therefore, the number of agent states must remain constant if a protocol is \emph{uniform}, \ie it does not assume any prior knowledge of the population size. However, solving some problems and/or designing fast protocols require a super-constant number of states. For example, Doty and Soloveichik \cite{DS18} proved that $\omega(1)$ states are necessary to solve the leader election problem within $o(n^2)$ expected interactions under the uniformly random scheduler.
Thus, several studies~\cite{BEF+21,DE19} have adopted a generalized model in which agents are defined as Turing machines, and the number of agent states is defined as $|\Sigma|^\ell$, where $\Sigma$ is the tape alphabet and $\ell$ is the maximum number of tape cells written by any agent in any execution of the protocol.
 \footnote{
 The \emph{expected} number of states can also be evaluated under the uniformly random scheduler, \eg in \cite{BEF+21}.
 }
We also adopt this generalized model for our uniform protocol $\CIG$, which does not assume any prior knowledge, yet utilizes $O(n^2)$ states.  If we want to adhere to the original model for our protocol $\CIG$, we require knowledge of a common upper bound $P$ on $n$. In this case, the protocol $\CIG_P$, now depending on the knowledge of $P$, identifies complete graphs using $O(P^2)$ states.
Here, we never use the knowledge $P$ except for specifying the number of states, that is, bounding the domain of two variables.

\subsection{Related Works}
Since Angluin, Aspnes, Diamadi, Fischer, and Peralta introduced this model in 2004 \cite{AngluinADFP2006}, various problems have been studied, including the leader election problem~\cite{Majority.space,AlistarPoly,AngluinADFP2006,Epidemic,BEF+21,BGK20,DS18,FastSpaceLE,GSU19,MICHAIL2022104698,SudoLB,sudotime}, the majority problem~\cite{AlistarhTradeoff,Majority.space,AlistarhFastExactMajority,SimpleMajority,BenMajority,berenbrink_et_al:LIPIcs.DISC.2018.10,BEF+21,BilkeMajority,Majority.optimal,MertziosMajority,MajorityCounting,MocMajority}, predicate computation \cite{AngluinADFP2006,AAE06,Epidemic,AAER07}, and the counting problem~\cite{AspnesBBS2016,count.space}. 
Most studies focus only on complete graphs, while several explore other graph classes, including arbitrary graphs \cite{FastGraph,NearArbGraph,AngluinADFP2006,MajorityArbitrary,yasumi_et_al:LIPIcs.OPODIS.2021.13,BipartitionYasumi}. There are also numerous studies on \emph{self-stabilizing} population protocols \cite{SSLEar2,CC19,CC20,SUDO2012100,LSLEonArbi,YOKOTA.time,Yokota.near}, initiated by Angluin, Aspnes, Fischer, and Jiang \cite{Angluin.ss}.

There are some variants of the population protocol model, such as the mediated population protocol model and the network construction model, which are closely related to the graph class identification problem. The mediated population protocol model is an extension of the population protocol model introduced by Michail, Chatzigiannakis, and Spirakis~\cite{MICHAIL20112434}. This model enables information to be stored not only in the states of agents but also in the communication links between them. The network construction model, an extension of the mediated protocol model, was introduced by Michail, and Spirakis~\cite{networkconstruction}. This model aims to construct specific networks based on the states of agents and the information carried by links between agents.

\section{Preliminaries}

\subsection{Population Protocols}\label{def:2.1}

The communication graph is represented by a simple and weakly connected digraph $G=(V,E)$,
where $V$ denotes the set of agents, and $E\subseteq \{(a,b) \in V \times V \mid a\neq b\}$ denotes the set of ordered pairs of agents that can interact: agents $u$ and $v$ can have interaction such that $u$ is the initiator and $v$ is the responder if $(u,v) \in E$. 
We say that $v$ is  an \emph{out-neighbor} of $u$ if $(u,v) \in E$.
At each time step, an adversarial scheduler selects one ordered pair $(u,v) \in E$.

As mentioned in Section \ref{sec:intro}, we adopt the original model for non-uniform protocols (\ie $\CIW_n$ and $\CIW_{n,k}$)
and adopt the generalized model given by \cite{DE19} for uniform and $\omega(1)$-space protocols.
In the original model, agents are defined as finite state machines. A protocol is defined as a 5-tuple $\mathcal{P}=(Q,\rho,Y,\pi,\delta)$, where $Q$ represents the set of agent states, $\rho \in Q$ is the initial state,
$Y$ the set of output symbols, $\pi:Q\rightarrow Y$ the output function, and $\delta:Q\times Q\rightarrow Q\times Q$ the state transition function. An agent outputs the symbol $\pi(s) \in Y$ when in state $s \in Q$.
Suppose agents $u$ and $v$ engage in an interaction, with $u$ as the initiator and $v$ as the responder, while in states $p$ and $q$ respectively. They then update their states to $p'$ and $q'$ respectively, where $(p',q')=\delta(p,q)$.
The number of states of a protocol $\mathcal{P}=(Q,\rho,Y,\pi,\delta)$ is simply defined as $|Q|$.
Roughly speaking, the generalized model, which accommodates both uniform and super-constant protocols, can be described as follows:
\begin{itemize}
    \item The definition of a protocol $\mathcal{P}=(Q,\rho,Y,\pi,\delta)$ is the same as in the original model
    except for the set $Q$ of agent states.
    \item An agent maintains a constant number of variables $x_1, x_2, \dots, x_s$, and the combination of their values constitutes the state of the agent. The first element $Q$ of the protocol is the set of all such states.
    The domain of each variable may not be bounded, thus $Q$ may be an infinite set.    
    \item For any agent $a \in V$, define $a$'s amount of information at time step $t$ as $f(a,t)=\sum_{i=1}^s \iota_i$, where each $\iota_i$ is the number of bits required to encode the value of variable $x_i$. The number of states of the protocol $\mathcal{P}$ is defined as $\max\{2^{f(a,t)} \mid a \in V, t=0,1,\dots\}$ in any execution of protocol $\mathcal{P}$. Therefore, the number of states of $\mathcal{P}$ may depend on the population size $n$, even if $\mathcal{P}$ itself does not depend on $n$.
\end{itemize}
Refer to Section 2 in \cite{DE19} for the formal definition of the generalized model, which is based on a Turing Machine.

A global state or a \emph{configuration}, denoted $C:V\to Q$, specifies the state of all agents.
The \emph{initial configuration}, where all agents are in the initial state $\rho$, is denoted by $\mathcal{I}$.
A configuration $C$ \emph{changes} to $C'$ via an interaction $(u,v) \in E$, denoted by $C \stackrel{(u,v)}{\to} C'$, if $(C'(u),C'(v)) = \delta(C(u),C(v))$ and $C'(w) = C(w)$ for all $w \in V \setminus \{u,v\}$. When $C$ changes to $C'$ via any interaction, we say that $C$ \emph{can change} to $C'$, denoted as $C\rightarrow C'$. An \emph{execution} is defined as an infinite sequence of configurations $\Xi = C_0, C_1, \dots$ that starts from the initial configuration $C_0=\mathcal{I}$ and satisfies $C_i \to C_{i+1}$ for every $i \in \mathbb{N}$. A configuration $C'$ is \emph{reachable} from $C$ if there exists a sequence of configurations $C_0, C_1, \dots, C_t$ with $C = C_0$, $C' = C_t$, and $C_i \to C_{i+1}$ for each $i$ from $0$ to $t-1$.
A configuration $C$ is \emph{stable} if, for any agent $a \in V$ and any configuration $C'$ reachable from $C$, agent $a$ outputs the same symbol in both $C$ and $C'$, \ie $\pi(C(a)) = \pi(C'(a))$.
We say that an execution \emph{stabilizes} when it reaches a stable configuration.

Fairness is defined as a predicate on executions, \ie it specifies the set of acceptable executions of a protocol $\mathcal{P}$.
An execution $\Xi =C_0,C_1,\dots$ is \emph{globally fair}
if for any configuration $C$ that appears infinitely often in $\Xi$,
every configuration reachable from $C$ also appears infinitely often in $\Xi$.
An execution $\Xi =C_0,C_1,\dots$ is \emph{weakly fair}
if there exists a sequence of interactions $\gamma = \gamma_0,\gamma_1,\dots$
such that every ordered pair $(u,v) \in E$ appears infinitely often in $\gamma$.

\subsection{Complete Graph Identification}
A communication graph $G=(V,E)$ is \emph{complete}
if $(u,v) \in E$ for any distinct agents $u,v \in V$.
A protocol $\mathcal{P}$ \emph{solves the complete graph identification} or \emph{identify complete graphs}
under global fairness (resp.~weak fairness)
if any globally fair (resp.~weakly fair) execution of $\mathcal{P}$
eventually reaches a stable configuration where
the outputs of all agents are $\yes$ if $G$ is complete, and $\no$ otherwise.

\subsection{Time Complexity}
The time complexity is defined as the number of interactions required to reach a stable configuration, under the assumption of a uniformly random scheduler that selects an ordered pair of agents uniformly at random at each time step.\footnote{
When focusing on complete graphs, time complexity is often expressed as \emph{parallel time}, which refers to the number of interactions required for stabilization divided by the number of agents. In this paper, we do not use parallel time because we are addressing complete graph identification on arbitrary graphs.
}
Given that we assume a uniformly random scheduler, the time complexity is considered a random variable; therefore, we evaluate it in terms of expectation and/or with high probability.
In this paper, \emph{with high probability} is defined as \emph{with probability $1-O(1/n)$}.

Introducing (asynchronous) \emph{rounds} can be useful
for analyzing the time complexities of protocols that work under weak fairness.

\begin{definition}[rounds]
Let $\gamma = \gamma_0, \gamma_1, \dots$ be an infinite sequence of interactions. 
The first round of $\gamma$ is the shortest prefix of $\gamma$ such that 
every $(u,v) \in E$ appears at least once in this prefix.
Let $t_1$ be the length of the first round of $\gamma$, 
\ie the prefix is $\gamma_0, \gamma_1, \dots, \gamma_{t_1-1}$.
For any $i \ge 1$, the $(i+1)$-th round of $\gamma$ is defined as the shortest prefix of the infinite sequence 
$\gamma_{S_i}, \gamma_{S_i+1}, \dots,$ where each pair $(u,v) \in E$ appears at least once in the prefix,
with $S_i = \sum_{k=1}^i t_k$ and $t_k$ being the length of the $k$-th round of $\gamma$.
\end{definition}

The standard analysis of the coupon collector's problem give the following lemma by considering each $(u,v) \in E$ as a type of coupon.
Note that
$|E| = \Omega(n)$ holds here by the weak connectivity of the communication graph.

\begin{lemma}
\label{lemma:round}
Let $\Gamma = \Gamma_0, \Gamma_1, \dots$ be the infinite sequence of interactions chosen by the uniformly random scheduler, where $\Pr(\Gamma_i=(u,v)) = 1/|E|$ for any $(u,v) \in E$ and $i \ge 0$, and this probability is independent of any other interaction $\Gamma_j\ (j \neq i)$.
Then, for any $i \ge 0$, the $i$-th round of $\Gamma$ has an expected length of $O(|E| \log |E|)=O(|E| \log n)$. Moreover, for any constant $c \ge 1$, the length is $O(c \cdot |E| \log n)$ with probability $1-O(n^{-c})$.
\end{lemma}

\section{Complete Graph Identification under Weak Fairness}
This section presents the protocol $\CIW_n$, which identifies complete graphs using the exact knowledge of $n$. It also introduces a generalized version, $\CIK$, which incorporates a design parameter $k$ where $1 \leq k \leq n$.

\begin{algorithm}[t]
\caption{$\CIW_n$ (a protocol for weak fairness with exact $n$)}
\label{al:n}
\setcounter{AlgoLine}{0}
\nonl\when{$a, b \in V$ interacts with $a$ as initiator and $b$ as responder}{    
    \uIf{$a \in \vl \land b \in \vl$}{        
        $a.\cnt \gets a.\cnt+b.\cnt$\;
        $(b.\leader,b.\cnt) \gets (F,0)$\;    
        \If{$a.\cnt=n$}{
            $(a.\phase,a.\cnt) \gets (2,0)$\;
        }
    }\uElseIf{$a \in \vl \cap V_2 \land a.\mode = b.\mode$}{
        $a.\cnt \gets a.\cnt+1$\;        
        $b.\mode \gets 1-b.\mode$\;                
        \If{$a.\cnt = n-1$}{
            $(a.\phase,a.\cnt,a.\mode)\gets(3,1,1-a.\mode)$\;
        }
    }\uElseIf{$a \in \vl \cap V_3 \land b \in V_1$}{
        $a.\leader \gets F$\;         
        $(b.\leader,b.\phase) \gets (L,2)$\;
    }        
    \uElseIf{$a \in V_3 \land b \in V_3 \land a.\cnt > 0 \land b.\cnt > 0$}{
        $a.\cnt \gets a.\cnt + b.\cnt$\;
        $b.\cnt \gets 0$\;
        \If{$a.\cnt = n$}{
            $a.\phase \gets 4$
        }
    }
    \ElseIf{$a \in V_4$}{
        $b.\phase \gets 4$\;
    }
}
\end{algorithm}
\subsection{Protocol \texorpdfstring{$\CIW_n$}{CIW\_n}}\label{sec:ciw}

In protocol $\CIW_n$, each agent $a \in V$ maintains four variables: $a.\leader \in \{F,L\}$, $a.\phase \in \{1,2,3,4\}$, $a.\mode \in \{0,1\}$, and $a.\cnt \in \{0,1,\dots,n\}$, resulting in exactly $16(n+1) = O(n)$ states per agent. Initially, $a.\leader = L$, $a.\phase = 1$, $a.\mode = 0$, and $a.\cnt = 1$. An agent $a$ is considered a \emph{leader} if $a.\leader = L$; otherwise, it is termed a \emph{follower}. The set of leaders is denoted by $\vl \subseteq V$. For each $i \in \{1,2,3,4\}$, an agent $a$ is said to be in phase $i$ if $a.\phase = i$, and the set of agents in phase $i$ is denoted by $V_i \in V$. Agents output $\yes$ if and only if they are in phase 4. The objective is to ensure that all agents reach phase 4 if the communication graph $G$ is complete 
(\ie the out-degrees of all agents are $n-1$), 
and that no agents reach phase 4 otherwise.

Our protocol $\CIW_n$, presented in Algorithm \ref{al:n}, can be summarized as follows:
\begin{itemize}
\item Agents in phase 1 participate in leader election (lines 1--5).  Once a unique leader is elected, it advances to phase 2.
\item The unique leader in phase 2 counts the number of its out-neighbors (lines 6--8). It moves to phase 3 upon determining that the number of its out-neighbors equals $n-1$ (lines 9--10). At the next interaction where this leader meets a follower in phase 1, the leadership status is transferred to the follower, who then moves to phase 2 and starts its own count (line 11--13). 
\item Agents in phase 3 count the number of agents in phase 3 (lines 14--16). If some agent observes that $|V_3| = n$, indicating that the out-degrees of all agents are $n-1$, it transitions to phase 4 (line 17--18).
\item
Once an agent in phase 4 appears, all agents proceed to phase 4 via the well-known one-way epidemic protocol \cite{Epidemic} (line 19--20):
In each interaction where an agent $a \in V_4$ is the initiator and an agent $b \in V_3$ is the responder, $b$ transitions to phase 4.
\end{itemize}

In what follows, we describe how to elect a leader, count the out-degree, and determine if $|V_3| = n$, using the variable $\cnt$. Note that $\cnt$ is utilized in phases 1, 2, and 3 for different purposes. This \emph{re-use} of $\cnt$ contributes to reducing the number of states of $\CIW_n$ from $O(n^3)$ to $O(n)$.

Leaders in phase 1 elect the unique leader.
Initially, all agents are leaders. When two leaders meet, one becomes a follower. Given that all agents know the exact value of $n$, the completion of the leader election is easily detected as follows:
(i) initially, all agents set $\cnt = 1$,
(ii) when two leaders interact, the surviving leader absorbs the $\cnt$ of the defeated leader,
(iii) the unique leader recognizes its uniqueness when $\cnt = n$.
Upon this detection, the unique leader advances to phase 2 with resetting $\cnt$ value to zero.

A leader in phase 2 determines if its out-degree is $n-1$. When a new leader $a \in V_2$ emerges, either from the leader election or the transfer of leadership status, $a.\cnt$ is  0, and all agents share the same $\mode$ value, either 0 or 1. Each time $a$, as the initiator, meets another agent $b$ with the same $\mode$ value, it increments $a.\cnt$ by one and toggles $b.\mode$ from 0 to 1 or from 1 to 0. Thus, $a.\cnt$ reaches $n-1$ if and only if $a$ meets $n-1$ different out-neighbors, implying that $a$'s out-degree is $n-1$. Then, it transitions to phase 3, setting $\cnt$ to 1 and toggling its $\mode$ value. Toggling $a.\mode$ ensures that all agents share the same $\mode$ value when $a$ meets a follower $b$ in phase 1 and transfers its leadership status to $b$, making $b$ a new leader in phase 2.

Agents in phase 3 determine whether $|V_3| = n$. As mentioned previously, every agent sets its $\cnt$ value to 1 upon transitioning to phase 3. When two agents in phase 3 with non-zero $\cnt$ values meet, one agent absorbs the $\cnt$ value of the other. The weak fairness ensures that all agents will meet each other infinitely often, leading eventually to a configuration where one agent in phase 3, denoted as $a$, has $a.\cnt = |V_3|$. If $a.\cnt = n$, then $a$ confirms that the out-degrees of all agents are exactly $n-1$, indicating that the communication graph is complete. Subsequently, $a$ advances to phase 4.

\begin{lemma}
\label{lemma:v34}
$\sum_{a\in V_3 \cup V_4} a.\cnt = |V_3 \cup V_4|$
always holds in an execution of $\CIW_n$.
\end{lemma}
\begin{proof}
This lemma follows directly from the facts that:
(i) each time an agent enters phase 3, it sets its $\cnt$ value to 1;
(ii) an agent in phase 3 updates its $\cnt$ value only during an interaction with another agent in phase 3, which does not alter the sum of the $\cnt$ values of those two agents;
and (iii) no agent in phase 4 updates its $\cnt$ value.
\end{proof}

% {nteiri}
\begin{theorem}\label{n_teiri}
Given the exact population size $n$, protocol $\CIW_n$ solves the complete graph identification problem under weak fairness, using $O(n)$ states.
\end{theorem}

\begin{proof}
While all agents are in phase 1, the invariant $\sum_{a \in V_1} a.\cnt = n$ always holds. Thus, only one leader may have $a.\cnt = n$ and transition to phase 2, with all agents having a zero $\cnt$ value. If the communication graph is complete, weak fairness guarantees that a unique leader is eventually elected. If not, the leader election may remain incomplete, but this does not matter to us. In both scenarios, once an agent transitions to phase 2, there is always exactly one leader.

Suppose that the communication graph is \emph{not} complete. Then, there exists at least one agent, say $b$, with an out-degree of at most $n-2$.
Even if $b$ transitions from phase 1 to phase 2, its $\cnt$ value is zero at that time. Each out-neighbor of $b$ can increase $b.\cnt$ by at most one, ensuring that $b.\cnt$ never reaches $n-1$. Therefore, $b$ does not advance to phase 3, and consequently, $|V_3| < n$ always holds. By Lemma \ref{lemma:v34}, this implies that no agent observes $\cnt = n$, and thus no agent transitions to phase 4. Therefore, every agent always outputs $\no$ from the beginning of an execution if the communication graph is not complete.

Next, suppose that the communication graph is complete.
As mentioned previously, the unique leader is eventually elected. 
When this leader, say $\ell$, is elected and moves to phase 2, 
its $\cnt$ value is zero and the $\mode$ values of all agents are zero. 
Weak fairness ensures that $\ell$ meets its $n-1$ out-neighbors as the initiator,
changing the $\mode$ values of all agents including $\ell$ to $1$,
and then moves to phase 3.
The leader $\ell$ gives the leadership status to another agent, say $\ell'$,
in the next interaction.
When $\ell'$ becomes a leader, $\ell'.\cnt=0$ and the $\mode$ values of all agents are 1.
Similarly, weak fairness ensures that $\ell'.\cnt$ reaches $n-1$, the $\mode$ values of all agents go back to $0$, and $\ell'$ transitions to phase 3.
Repeating this process, all agents eventually enter phase 3.
By Lemma \ref{lemma:v34}, weak fairness ensures that the event that some agent $a$ in phase 3 observes $a.\cnt = n$ eventually occurs, and $a$ enters phase 4.
After that, all agents enter phase 4 by meeting $a$ or another agent in phase 4,
by weak fairness.
\end{proof}

\begin{lemma}
\label{lemma:incomplete}
If the communication graph is not complete,
every execution of $\CIW_n$ stabilizes in zero time,
\ie the initial configuration $\mathcal{I}$ is stable.
\end{lemma}

\begin{proof}
As shown in the proof of Theorem \ref{n_teiri},
no agent enters phase 4 if the communication graph is not complete,
which gives the lemma.
\end{proof}

\begin{lemma}
\label{lemma:complete}
If the communication graph is complete, every execution of $\CIW_n$ stabilizes in $O(n)$ rounds under weak fairness.
\end{lemma}
\begin{proof}
The leader election will be completed before every pair of agents meets, thus finishing in one round. Each agent requires one round to recognize that its out--degree is $n-1$ and advance to phase 3. Simultaneously, it takes one round for a leader in phase 3 to encounter an agent in phase 1 and transfer leadership status. Therefore, all agents reach phase 3 within $2n+1$ rounds. An additional round is sufficient for some agent to detect that $|V_3|=n$ and move to phase 4, after which all agents enter phase 4 in one subsequent round.
\end{proof}

The following theorem is derived from Lemmas \ref{lemma:round}, \ref{lemma:incomplete}, and \ref{lemma:complete}.
\begin{theorem}\label{theorem:ntime}
Any execution of $\CIW_n$ stabilizes in $O(n^3 \log n)$ steps, both in expectation and with high probability,
under the uniformly random scheduler.
\end{theorem}

\begin{algorithm}[!htb]
\caption{A protocol $\CIK$}
\label{al:k}
\setcounter{AlgoLine}{0}
\nonl\when{$a, b \in V$ interacts with $a$ as initiator and $b$ as responder}{
    \uIf{$a \in \vl \land b \in \vl$}{        
        $a.\cnt \gets a.\cnt + b.\cnt$\;
        $(b.\leader,b.\cnt) \gets (F,0)$\;    
        \If{$a.\cnt = n$}{
            $a.\phase \gets 1.5$\;
        }
    }\uElseIf{$a.\leader \in \vl \cap V_{1.5} \land b.\group = k$}{
        $a.\cnt \gets a.\cnt - 1$\;
        $b.\group \gets a.\cnt \bmod k$\;
        \If{$a.\cnt < k$}{
            $(b.\leader, b.\phase) \gets (L, 2)$\;
        }
        \If{$a.\cnt = 1$}{
            $(a.\leader, a.\phase, a.\cnt, a.\group) \gets (L, 2, 0, 0)$\;
        }
    }\uElseIf{$a \in \vl \cap V_2 \land a.\mode[a.\group] = b.\mode[a.\group]$}{
        $a.\cnt \gets a.\cnt+1$\;        
        $b.\mode[a.\group] \gets 1-b.\mode[a.\group]$\;                
        \If{$a.\cnt = n-1$}{
            $(a.\phase,a.\cnt,a.\mode[a.\group])\gets(3,1,1-a.\mode[a.\group])$\;
        }
    }\uElseIf{$a \in \vl \cap V_3 \land b \in V_1 \land a.\group = 
    b.\group$}{
        $a.\leader \gets F$\;
        $(b.\leader,b.\phase) \gets (L,2)$\;
    }        
    \uElseIf{$a \in V_3 \land b \in V_3 \land a.\cnt > 0 \land b.\cnt > 0$}{
        $a.\cnt \gets a.\cnt + b.\cnt$\;
        $b.\cnt \gets 0$\;
        \If{$a.\cnt = n$}{
            $a.\phase \gets 4$
        }
    }
    \ElseIf{$a \in V_4$}{
        $b.\phase \gets 4$\;
    }
}
\end{algorithm}

\subsection{Protocol \texorpdfstring{$\CIK$}{CIW\_{n,k}}}
\label{sec:cik}
This section presents $\CIK$, a generalized version of $\CIW_n$ that incorporates a design parameter $k$, where $1 \leq k \leq n$.

The time complexity of $\CIW_n$ is primarily determined by the degree recognition process, where $n$ agents sequentially determine whether their out-degrees are $n-1$, requiring $O(n)$ rounds. We reduce the duration of this period by dividing the population into $k$ groups, nearly equally.
Specifically, after the unique leader $\ell$ is elected, $\ell$ assigns group identifiers to all agents as follows:
\begin{itemize}
\item Each agent $a$ maintains a variable $a.\group \in \{0,1,\dots,k\}$ to indicate the group to which it belongs. The initial value of this variable is $k$, which corresponds to the \emph{null} group.
\item When $\ell$ becomes the unique leader, $\ell.\cnt = n$ holds. Thereafter, whenever $\ell$ encounters an agent $b$ with $b.\group = k$, $\ell$ decrements $\ell.\cnt$ by one and assigns $(\ell.\cnt \bmod k)$ to $b.\group$.
In parallel, if $\ell.\cnt < k$ after decrementing, $\ell$ reverts the opposing agent $b$ to a leader, \ie setting $b.\leader$ to $L$.
\item When $\ell.\cnt$ reaches $1$, it recognizes that it has already assigned group identifiers to all agents except itself. Then, $\ell$ assigns $0$ to its $\group$ and completes the process of assigning group identifiers.
\end{itemize}
Through this process, the population is divided into $k$ groups, each consisting of either $\lfloor n/k\rfloor$ or $\lceil n/k \rceil$ agents, with exactly one leader per group. The leader of each group $g \in \{0,1,\dots,k-1\}$ counts the number of its out-neighbors in the same manner as $\CIW_n$, but in parallel. Unlike $\CIW_n$, in protocol $\CIK$, the variable $a.\mode$ is an array of size $k$: $a.\mode[g] \in \{0,1\}$ is used for counting the number of out-neighbors in group $g$. The remaining parts of the protocols $\CIW_n$ and $\CIK$ are identical.

% Since $\mode$ is the binary array with size $k$, variable $a.\mode$ requires $2^k$ states, 
% while $\cnt$ and $\group$ require $n+1$ and $k+1$ states, respectively.
% All other variables requires only $O(1)$ states.
% Thus, $\CIK$ requires $O(nk2^k)$ states per agent.
% At the cost of increased states, $\CIK$ reduces the time complexity from $O(n^3 \log n)$
% to $O((n^3/k) \log n)$ steps both in expectation and with high probability because
% each group consists of at most $\lceil n/k \rceil$ agents, thus $O(n/k)$ rounds are sufficient
% to complete the out-degree counting process of all agents. 

In protocol $\CIW_{n,k}$, each agent $a\in V$ maintains five variables:
$a.\leader \in \{F,L\}$, $a.\phase \in \{1,1.5,2,3,4\}$, $a.\mode \in \{0,1\}^{k}$, $a.\group \in \{0,1,...,k\}$, and $a.\cnt \in \{0,1,...,n\}$, resulting in exactly $10(n+1)(k+1)2^k=O(nk2^k)$ states per agent.
Initially, $a.\leader = L$, $a.\phase = 1$, $a.\mode[i] = 0 (i\in[0,k-1])$, $a.\group=k$, $a.\cnt = 1$.
Note that all of elements of $a.\mode$ are $0$ since $a.\mode$ is a binary array with size $k$.

The roles of values, output, and set of definitions are the same as in $\CIW_n$, with the following exceptions:
(i) There is a new phase, 1.5. We denote the set of agents in phase 1.5 by $V_{1.5} \subset V$, where an agent $a$ is in phase 1.5 if $a.\phase = 1.5$.
(ii) Each group has a group leader. An agent $a$ is the leader of group $i$ if $a \in \vl \land a.\group = i$.

Our protocol $\CIW_{n,k}$, presented in Algorithm~\ref{al:k}, can be summarized as follows:
\begin{itemize}
    \item Agents in phase 1 participate in leader election (lines 1--5).
    Once a unique leader $\ell$ is elected, it advances to phase 1.5.
    \item In phase 1.5, $\ell$ divides all agents into $k$ uniform groups (lines 6--8) and appoints a unique group leader $i\in[0,k-1]$ for each group (lines 9--12). Once each group leader is appointed, the leader advances to phase 2.
    \item The group $i$ leader in phase 2 counts the number of its out-neighbors (lines 13--15). It moves to phase 3 upon determining that the number of its out-degrees equals $n-1$ (lines 16--17). At the next interaction where this group leader meets a same group follower in phase 1, the leadership status is transferred to the follower, who then moves to phase 2 and begins their own count (line 18--20).
    \item Agents in phase 3 count the number of agents in phase 3 (lines 21--23). If some agent observes that $|V_3| = n$, indicating that the out--degrees of all agents are $n-1$, it transitions to phase 4 (line 24--25).
    \item Once an agent in phase 4 appears, all agents proceed to phase 4 via the well-known one-way epidemic protocol \cite{Epidemic} (line 26--27):
    In each interaction where an agent $a \in V_4$ is the initiator and an agent $b \in V_3$ is the responder, $b$ transitions to phase 4.
\end{itemize}

In what follows, we describe how to elect a leader, divide into $k$ groups, count the out--degree, and determine if $|V_3|=n$, using the variable $\cnt$.
Note that $\cnt$ is utilized in phase 1, 1.5, 2, and 3 for different purposes.
That re--use of $\cnt$ contributes to reducing the number of states in $\CIW_{n,k}$ from $O(n^4k2^k)$ to $O(nk2^k)$.

Leader election is almost same as $\CIW_n$.
The only difference is that when a unique leader advances to phase 2, it resets $\cnt$ to $n$ rather than to zero.

A unique leader $\ell$ in phase 1.5 divides all agents into $k$ uniform groups.
Initially all agents' $\group$ values are set to $k$, indicating that they have not yet been assigned to a group.
When $\ell$ interacts with an agent whose $\group$ is $k$, $\ell$ decreases its $\cnt$ by 1 and set the responder's $\group$ to $\ell.\cnt \bmod k$.
If $\ell.\cnt$ is less than $k$, the responder becomes a new group leader and advances to phase 2.
When $\ell.\cnt$ equals to one, $\ell$ set its own $\group$ to zero, and advances to phase 2, resetting $\ell.\cnt$ to zero.
At this stage, the group leader 0 is $\ell$, ensuring there is exactly one leader in each group and the uniform partition into $k$ groups is completed.

A group $i$ leader in phase 2 determines if its out--degree is $n-1$.
When a new group $i$ leader $a \in V_2$ emerges, either from the $k$ partition or from the transfer of leadership status, $a.\cnt$ is 0, and all agents share the same $\mode[i]$ value, either 0 or 1.
Each time $a$, as the initiator, meets another agent $b$ with the same $\mode[i]$ value, it increments $a.\cnt$ by one and toggles $b.\mode$ from 0 to 1 or from 1 to 0. 
Thus, $a.\cnt$ reaches $n-1$ if and only if $a$ meets $n-1$ different out-neighbors, implying that $a$'s out-degree is $n-1$. 
Then, it transitions to phase 3, setting $\cnt$ to 1 and toggling its $\mode[i]$ value. 
Toggling $a.\mode[i]$ ensures that all agents share the same $\mode[i]$ value when $a$ meets a same group follower $b$ in phase 1 and transfers its leadership status to $b$, making $b$ a new group $i$ leader in phase 2.

Agents in phase 3 determine whether $|V_3|=n$.
This process is exactly same as in $\CIW_n$.

\begin{lemma}\label{lemma:kv34}
$\sum_{a\in V_3 \cup V_4} a.\cnt = |V_3 \cup V_4|$
always holds in any execution of $\CIW_{n,k}$.
\end{lemma}
\begin{proof}
The proof is identical to that of Lemma~\ref{lemma:v34}.
\end{proof}

\begin{theorem}\label{k_teiri}
Given the exact population size $n$ and a design parameter $k$ with $1 \le k \le n$,
protocol $\CIK$ solves the complete graph identification problem under weak fairness, using $O(n k 2^k)$ states.
\end{theorem}

\begin{proof}
While all agents are in phase 1, the invariant $\sum_{a \in V_1} a.\cnt = n$ always holds. Thus, only one leader may have $a.\cnt = n$ and transition to phase 1.5, with all follower agents having a zero $\cnt$ value. 
If the communication graph is complete, weak fairness guarantees that a unique leader is eventually elected. 
If not, the leader election may remain incomplete, but this does not matter to us. In both scenarios, once an agent transitions to phase 1.5, there is always exactly one leader.

Suppose that the communication graph is \emph{not} complete. 
Then, there exists at least one agent, say $b$, with an out-degree of at most $n-2$.
Even if $b$ transitions from phase 1 to phase 2 or phase 1.5 to phase 2, its $\cnt$ value is zero at that time. 
Each out-neighbor of $b$ can increase $b.\cnt$ by at most one, ensuring that $b.\cnt$ never reaches $n-1$. 
Therefore, $b$ does not advance to phase 3, and consequently, $|V_3| < n$ always holds. 
By Lemma \ref{lemma:kv34}, this implies that no agent observes $\cnt = n$, and thus no agent transitions to phase 4. 
Therefore, every agent always outputs $\no$ from the beginning of an execution if the communication graph is not complete.

Next, suppose that the communication graph is complete.
As mentioned previously, the unique leader is eventually elected. 
When this leader, say $\ell$, is elected and moves to phase 1.5 with setting $\cnt$ to $n$.
Weak fairness ensures that $\ell$ meets its $n-1$ out--neighbors as the initiator.
Thus, all follower agents are assigned to some groups.
By periodicity of remainder, the number of each groups' agents is either $\lfloor n/k\rfloor$ or $\lceil n/k \rceil$.
Therefore, uniform $k$ partition is completed.
$\ell$ makes a new group leader if and only if $\ell.\cnt$ less than $k$.
Thus, in the end of phase 1.5, each group has a exact one leader.
When a group $i$ leader in phase 2, say this leader $\ell_i$,
its $\cnt$ value is zero and the $\mode[i]$ values of all agents are zero. 
Weak fairness ensures that $\ell_i$ meets its $n-1$ out-neighbors as the initiator,
changing the $\mode[i]$ values of all agents including the group $i$ leader to $1$,
and then moves to phase 3.
The leader $\ell_i$ gives the leadership status to another same group agent, say $\ell'_i$,
in the next interaction.
When $\ell'_i$ becomes a leader, $\ell'.\cnt=0$ and the $\mode[i]$ values of all agents are 1.
Similarly, weak fairness ensures that $\ell'_i.\cnt$ reaches $n-1$, the $\mode[i]$ values of all agents go back to $0$, and $\ell'_i$ transitions to phase 3.
Repeating this process, all agents eventually enter phase 3.
Since the group $i$ only use $\mode[i]$, no interference of $\mode$ happens in parallel execution of out--degree counting. 
By Lemma \ref{lemma:kv34}, weak fairness ensures that the event that some agent $a$ in phase 3 observes $a.\cnt = n$ eventually occurs, and $a$ enters phase 4.
After that, all agents enter phase 4 by meeting $a$ or another agent in phase 4,
by weak fairness.
\end{proof}

\begin{lemma}
\label{lemma:k:incomplete}
If the communication graph is not complete,
every execution of $\CIW_{n,k}$ stabilizes in zero time,
\ie the initial configuration $\mathcal{I}$ is stable.
\end{lemma}

\begin{proof}
As shown in the proof of Theorem \ref{k_teiri},
no agent enters phase 4 if the communication graph is not complete,
which gives the lemma.
\end{proof}

\begin{lemma}\label{lemma:k:complete}
If the communication graph is complete, every execution $\Xi$ of $\CIW_{n,k}$ stabilizes in $O(n/k)$ rounds under weak fairness.  
\end{lemma}

\begin{proof}
The leader election will be completed before every pair of agents meets, thus finishing in one round.
The $k$-partition will be completed once the unique leader has met all other agents, also completing in one round.
Thereafter, in every two rounds, at least one agent in each group recognizes that its out-degree is $n-1$, advances to phase 3,
and transfers the leadership status to another agent within the same group who is in phase 1 (if one exists).
Therefore, all agents reach phase 3 within $2\lceil n/k \rceil+2$ rounds.
An additional one round is sufficient for some agent to detect that $|V_3|=n$ and move to phase 4, after which all agents enter phase 4 in one subsequent round.
\end{proof}

The following theorem is derived from Lemma~\ref{lemma:round}, ~\ref{lemma:k:incomplete}, and ~\ref{lemma:k:complete}.

\begin{theorem}\label{kjikan}
Any execution of $\CIK$ stabilizes in $O((n^3/k) \log n)$ steps, both in expectation and with high probability,
under the uniformly random scheduler.
\end{theorem}

\begin{algorithm}[htb]
\caption{A protocol $\CIG$}
\label{al:P}
\setcounter{AlgoLine}{0}
\Fn{$\RReset(a)$}{
    $(a.\leader, a.\phase, a.\mode, a.\cnt) \gets (L, 1, 0, 1)$\;
}
\nonl\when{ $a, b \in V$ interacts with $a$ as initiator and $b$ as responder}{
    \uIf{$a.\token = \top \land b.\token = \top$}{
        $b.\token \gets \bot$\;
        $a.\sz \gets b.\sz \gets a.\sz + b.\sz$\;
        $\RReset(a)$\;
        $\RReset(b)$\;
    }\uElseIf{$\exists (x,y)\in\{(a,b),(b,a)\}: x.\token = \top \land x.\sz \le y.\sz$}{
        $x.\token \leftrightarrow y.\token$\;
        $x.\sz \leftrightarrow y.\sz$\;
    }\ElseIf{$\exists (x,y)\in\{(a,b),(b,a)\}: x.\sz > y.\sz$}{
        $x.\token \leftrightarrow y.\token$\;
        $y.\sz \gets x.\sz$\;
        $\RReset(y)$\;
    }
    \If{$a.\sz = b.\sz$}{
        Execute protocol $\CIW_{a.\sz}$.\;
        If $\cnt$ exceeds $\sz$ during the execution of $\CIW_{a.\sz}$, set $\cnt$ to $\sz$.\;
    }
}
\nonl The notation $x \leftrightarrow y$ means that the values of $x$ and $y$ are swapped.
\end{algorithm}

\section{Complete Graph Identification under Global Fairness}\label{sec:cig}
This section proposes a protocol $\CIG$ for identifying complete graphs with no initial knowledge under global fairness using $32n(n+1)=O(n^2)$ states. 

This protocol uses $\CIW_n$ as a submodule. 
However, we do not assume the knowledge of exact $n$ here, whereas $\CIW_n$ requires knowledge of $n$ in advance. 
In $\CIG$, the agents try to compute the number of agents $n$
and store the estimated value of $n$ on a variable $\sz$.
As we will see later, for any agent $a \in V$, 
$a.\sz$ is monotonically non-decreasing and eventually reaches $n$.
Whenever two agents with the same $\sz$ value meets,
they execute $\CIW_{\sz}$, \ie updates the four variables $\leader$, $\phase$, $\cnt$, and $\mode$ according to
the transition function of $\CIW_{\sz}$.
Each time $a.\sz$ is updated, all variables of $\CIW_n$ in agent $a$ is reset by their initial values.

We compute the number of agents as follows:
\begin{itemize}
    \item An agent $a$ maintains two variables: $a.\token \in \{\bot,\top\}$ and $a.\sz \in \mathbb{N}$. Initially, $a.\token = \top$ and $a.\sz = 1$. An agent $a$ is said to have a token if $a.\token = \top$.
    \item When two agents $a$ and $b$, each with a token, interact, the two tokens merge. This is done by setting $b.\token$ to $\bot$ and simultaneously increasing $a.\sz$ by $b.\sz$. No other rules update the $\sz$ value of an agent with a token. Therefore, the invariant $\sum_{a \in V, a.\token = \top} a.\sz = n$ always holds.
    This token-merge process guarantees that the population eventually reaches a configuration where exactly one agent retains a token, denoted as $a_T$. Thereafter, $a_T.\sz = n$ always holds, thanks to the aforementioned invariant.
    \item Whenever an agent $a$ without a token encounters another agent $b$ with $a.\sz < b.\sz$, $a.\sz$ is updated to $b.\sz$. Thus, the maximum value $a_T.\sz = n$ eventually propagates to the entire population. Consequently, all agents will have $\sz = n$, although they do not know whether this value equals to the actual number of nodes, meaning they cannot detect the termination of the size computation.
\end{itemize}

In the protocol $\CIG$, each agent $a\in V$ maintains six  variables:
$a.\leader \in \{F,L\}$, $a.\phase \in \{1,2,3,4\}$, $a.\mode \in \{0,1\}$, $a.\cnt \in \mathbb{N}$, $a.\token \in \{\bot, \top\}$, and $a.\sz \in \mathbb{N}$.
As we discuss later, for any agent $a\in V$, it is always guaranteed that $a.\cnt \le n$ and $a.\sz \le n$.
Thus, each agent utilizes $32n(n+1)$ states.
Initially, $a.\leader = L$, $a.\phase = 1$, $a.\mode = 0$, $a.\cnt = 1$, $a.\token = \top$, and $a.\sz = 1$.
The four variables $a.\leader$, $a.\phase$, $a.\mode$, and $a.\cnt$ are the same as those used in $\CIW_n$, and we refer to these as \emph{$\CIW$'s variables}.
An agent $a$ is said to have a \emph{token} if $a.\token = \top$; otherwise, it is termed that $a$ does not have a token.
The \emph{token size} refers to the $\sz$ of the agent holding the token.
Each agent outputs $\yes$ if and only if the agent is in phase 4.

It is worth mentioning that, unlike $\CIW_n$, during an execution of $\CIG$, an agent $a$ may enter phase 4 even if the communication graph $G$ is \emph{not} complete. However, this false positive error occurs only if $a$ underestimates the population size. Agent $a$ resets all of $\CIW$'s variables to their initial values whenever the estimated population size, $a.\sz$, is updated. Therefore, the objectives here are to ensure that:
(i) every agent $a$ eventually has $a.\sz = n$,
(ii) every agent with $a.\sz = n$ reaches phase 4 if the communication graph is complete,
(iii) no agent with $a.\sz = n$ reaches phase 4 if the communication graph is not complete.

Our protocol $\CIG$, presented in Algorithm~\ref{al:P}, can be summarized as follows.
\begin{itemize}
    \item Initially, all agents possess one token each, with a size of 1. 
    Tokens (and their $\sz$ values) move randomly through agent interactions (lines 4,9, and 12).
    When two tokens meet during an interaction, the responder discards its token, and the initiator's $\sz$ absorbs the responder's $\sz$ (lines 2--5).
    \item When a token's size reaches $n$ for the first time, this size information is propagated to all agents.
    Subsequently, all agents reset their $\CIW$'s variables to the initial state and execute protocol $\CIW_n$ (line 16). However, as agents do not know the value of $n$ in this protocol, they cannot directly determine when its $\sz$ value reaches $n$.
    Therefore, each agent $a$ resets its $\CIW$'s variables to the initial values whenever $a.\sz$ is updated (\ie increments) (lines 3--14).
    Each time two agents meet, they execute $\CIW_n$ if and only if their sizes are equal (lines 15--16).
\end{itemize}

In what follows, we describe how to count the population size, and explain why the $\CIG$ protocol operates correctly.
It is important to note that $\sz$ is used to denote two different concepts: the size of an agent and the size of a token.
When agents possess tokens, it becomes necessary to account for both the agent's size and the token's size.
The initial counting process requires $O(n^2)$ states, but we reduce this to $O(n)$ by maintaining only the token's size or the agent's size, leveraging the fact that the agent's size and the token's size are equal when agents possess tokens.

Tokens move randomly through interactions as described below.

When both the initiator $a$ and the responder $b$  have tokens, the responder's token is disappeared.
Both agents' sizes are set to $a.\sz + b.\sz$, and both $\CIW$'s variables are reset to the initial state (lines 3--7).

When only one of the agents, either the initiator or the responder, holds a token, we denote the agent holding the token as $x$ and the agent without a token as $y$. 
If $x.\sz$ is less than or equal to $y.\sz$, the token moves from $x$ to $y$, and their sizes are swapped (lines 8--10).
During this process, agents do not reset their values.
% This ensures that for any agent $a\in V$, $a.\cnt \le n$ is always guaranteed due to resetting.
When $x.\sz$ is greater than $y.\sz$, the token moves from $x$ to $y$, $y.\sz$ is set to $x.\sz$, and $y$'s $\CIW$'s variables are reset to the initial state (lines 11--14).

When neither the initiator $a$ nor the responder $b$ has tokens, the following actions occur: If $a.\sz$ is larger than $b.\sz$, then $b.\sz$ is set to $a.\sz$ and $b$'s $\CIW$'s variables are reset to the initial state. Conversely, if $b.\sz$ is larger than $a.\sz$, then $a.\sz$ is set to $b.\sz$ and $a$'s $\CIW$'s variables are reset to the initial state (lines 11--14).

Eventually, only one agent retains a token, the size of the token becomes $n$, and all other agents' sizes also become $n$.

When the sizes of the initiator $a$ and the responder $b$ are equal, the agents execute the protocol $\CIW_{a.\sz}$ (lines 15--17).
If $\cnt$ exceeds $\sz$ during the execution of $\CIW_{a.\sz}$, set $\cnt$ to $\sz$.

\begin{theorem}\label{P_teiri}
Protocol $\CIG$ solves the complete graph identification problem under global fairness, using $O(n^2)$ states.
\end{theorem}

\begin{proof}
Global fairness ensures that tokens can visit all agents, allowing tokens to encounter each other and eventually consolidate into one. When tokens meet, the initiator's token absorbs the responder's token (\ie the initiator's $\sz$ absorbs the responder's $\sz$). The token size remains constant unless absorption occurs. Initially, as token sizes are $1$, when the number of tokens becomes one, the size of the token becomes $n$.

An agent's size changes if and only if an agent without a token interacts with another whose size is larger than its own. If the maximum token size is $k$, then all agents' sizes are less than or equal to $k$. Therefore, when the number of tokens reduces to one, all agents' sizes are less than or equal to $n$. 
Information about the population size propagates to all agents, thereby triggering a reset.

Since there is only one token in the population, token and agent sizes do not change after the number of tokens reduces to one. The protocol $\CIW_n$ runs if and only if an initiator's size and a responder's size are equal. Therefore, an agent whose size is $n$ never interacts with an agent whose size is less than $n$. Consequently, all agents' sizes equal $n$, and $\CIW_n$ operates under correct assumptions.
Protocol $\CIW_n$ runs correctly according to Theorem~\ref{n_teiri}, thereby ensuring $\CIG$ operates effectively.

Additionally, it always holds that $\forall a\in V:a.\sz\le n$ in any configuration of $\CIG$, and $\forall a\in V:a.\cnt\le n$. The former is proved as described above. From protocol $\CIG$, it always holds that $a\in V:a.\cnt \le a.\sz$, hence the latter is also true. Thus, $\CIG$ uses $O(n^2)$ states.
\end{proof}

\begin{lemma}\label{lemma:P:incomplete}
If the communication graph is not complete, every execution of $\CIG$ stabilizes in $O(n\hit_{\hat{G}}\log{n})$ steps, both in expectation and with high probability, under the uniformly random scheduler. 
\end{lemma}

\begin{proof}
After the size counting process has completed and $\sz = n$ has been propagated to all agents, all agents output $\no$ and their output does not change, as established in Lemma~\ref{lemma:incomplete}.

Sudo, Shibata, Nakamura, Kim, and Masuzawa~\cite{SudoGlobal} showed that the number of steps required for all tokens to merge into one is $O(n \hit_{\hat{G}} \log{n})$ in expectation (Lemma 4). They assume that the communication graph is undirected, \ie $(u,v) \in E$ implies $(v,u) \in E$ for any $u,v \in E$. However, this lemma does not depend on this assumption, thus it is applicable in our context. Moreover, the upper bound $O(n \hit_{\hat{G}} \log{n})$ also holds with high probability.
In fact, their proof for ``Lemma 4'' essentially proves that for any $c \ge 1$, all tokens merge into one within $O(c \cdot n \cdot \hit_{\hat{G}} \log{n})$ with probability $1-\binom{n}{2} \cdot n^{-c}=1-O(n^{2-c})$. Thus, by setting $c = 3$, we can easily establish the ``with high-probability'' bound.

Alistarh, Rybicki, and Voitovych~\cite{NearArbGraph} provided an upper bound on the time complexity required for the information of one agent to propagate to all agents on arbitrary graphs. On an arbitrary graph $G$, the worst-case time is $O(m(\log{n} + D))$ steps in expectation and $O(m(\log n + D) \log n)$ steps with high probability.

Therefore, the process of counting sizes and spreading $\sz = n$ to all agents occurs within $O(n \hit_{\hat{G}} \log{n})$ steps, both in expectation and with high probability.
\end{proof}

\begin{lemma}\label{lemma:P:complete}
If the communication graph is complete, every execution of $\CIG$ stabilizes in $O(n^3\log{n})$ steps, both in expectation and with high probability, under the uniformly random scheduler. 
\end{lemma}

\begin{proof}
Due to the symmetry of complete graphs, the process of tokens consolidating into one is analogous to leader election. Consequently, the size counting requires $O(n^2)$ steps. The spreading of $\sz = n$ to all agents is similar to the propagation of an infection, which takes $O(n \log{n})$ steps. Therefore, according to Theorem~\ref{theorem:ntime}, $\CIG$ stabilizes in $O(n^3 \log{n})$ steps, both in expectation and with high probability.
\qed
\end{proof}

The following theorem is derived from Lemma~\ref{lemma:P:incomplete},and ~\ref{lemma:P:complete}.

\begin{theorem}\label{Pjikan}
Any execution of $\CIG$ stabilizes in $O((\hit_{\hat{G}}+n^2)n \log n)$ steps both in expectation and with high probability under the uniformly random scheduler.    
\end{theorem}

\section{Impossibility}
Yasumi \etal~\cite{yasumi_et_al:LIPIcs.OPODIS.2021.13} introduced the transformed graph $f(G)=(V',E')$
for any communication graph $G=(V,E)$.
\begin{definition}
For any digraph $G=(V,E)$ with $V = \{v_1, v_2, ..., v_n\}$,
the digraph $f(G)=(V',E')$ is defined as follows:
$ V' = \{v_1', v_2', ..., v_{2n}'\}$
and $E' = \{(v_x', v_y'), (v_{x+n}', v_{y+n}') \in V' \times V' | (v_x, v_y) \in E\} \cup \{(v_1', v_{z+n}'), (v_{1+n}', v_z') \in V' \times V' | (v_1, v_z) \in E\}
\cup \{(v_{z+n}',v_1'), (v_z',v_{1+n}') \in V' \times V' | (v_z,v_1) \in E\}
$.
\end{definition}
They proved that for any simple and connected directed graph $G$, no protocol can distinguish between $G$ and $f(G)$ (Lemma 5 in \cite{yasumi_et_al:LIPIcs.OPODIS.2021.13}).\footnote{
They prove this fact for any \emph{undirected graph} $G=(V,E)$, where ``undirected'' means that $(u,v) \in E$ holds if and only if $(v,u) \in E$ for any $u, v \in V$. However, their proof also applies to general directed graphs without any modifications.
}
Specifically, they showed that for any protocol $\mathcal{P}$, there exists a weakly-fair execution $\Xi'$ on $f(G)$ that stabilizes to output $\yes$ (respectively, $\no$) if there is a weakly-fair execution $\Xi$ of $\mathcal{P}$ on $G$ that stabilizes to output $\yes$ (respectively, $\no$). Here, the statement that an execution stabilizes to output $x$ means that the execution reaches a stable configuration where all agents output $x$.

\begin{figure}[!htb]
\centering
\includegraphics[scale=0.3]{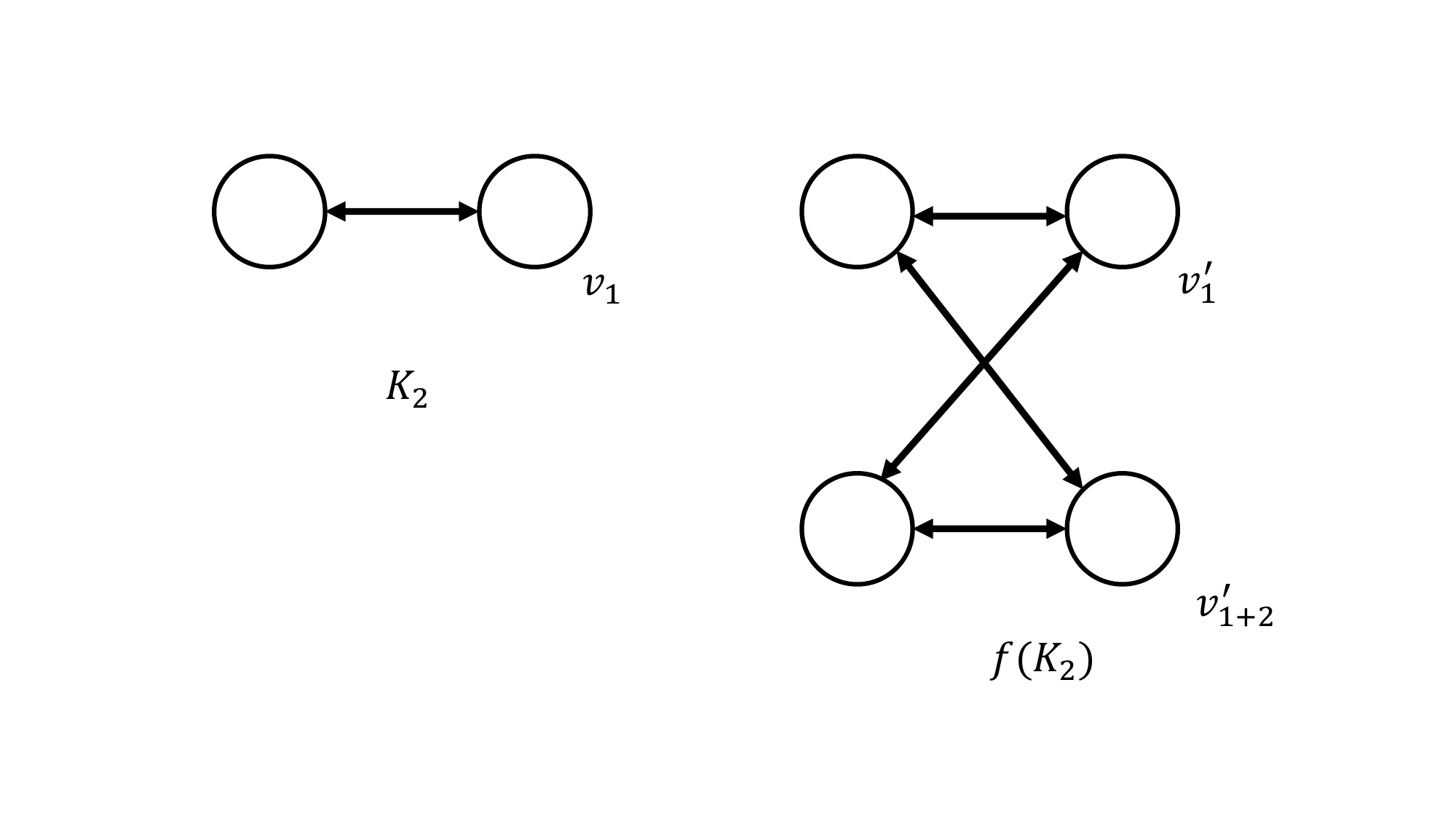}
\caption{Graphs $K_2$ and $f(K_2)$}
\label{y26}
\end{figure}

For any complete directed graph $K_n$ with size $n$, where $n \geq 2$, $f(K_n)$ is not complete.
(See Figure \ref{y26} for the case $n=2$.)
Consequently, the following theorem directly follows from the aforementioned lemma presented by Yasumi \etal~\cite{yasumi_et_al:LIPIcs.OPODIS.2021.13}.

\begin{theorem}\label{impossible:P}
There is no protocol that identifies complete graphs with initial knowledge of $P$ under weak fairness.
\end{theorem}

\clearpage

\bibliographystyle{plain}
\bibliography{short}

\begin{thebibliography}{10}

\bibitem{AlistarhTradeoff}
Dan Alistarh, James Aspnes, David Eisenstat, Rati Gelashvili, and Ronald~L. Rivest.
\newblock Time-space trade-offs in population protocols.
\newblock In {\em SODA}, pages 2560--2579, 2017.

\bibitem{Majority.space}
Dan Alistarh, James Aspnes, and Rati Gelashvili.
\newblock Space-optimal majority in population protocols.
\newblock In {\em SODA}, page 2221–2239, 2018.

\bibitem{AlistarPoly}
Dan Alistarh and Rati Gelashvili.
\newblock Polylogarithmic-time leader election in population protocols.
\newblock In {\em ICALP}, pages 479--491, 2015.

\bibitem{FastGraph}
Dan Alistarh, Rati Gelashvili, and Joel Rybicki.
\newblock Fast graphical population protocols.
\newblock In {\em OPODIS 2021}, pages 14:1--14:18, 2022.

\bibitem{AlistarhFastExactMajority}
Dan Alistarh, Rati Gelashvili, and Milan Vojnovi\'{c}.
\newblock Fast and exact majority in population protocols.
\newblock In {\em PODC}, page 47–56, 2015.

\bibitem{NearArbGraph}
Dan Alistarh, Joel Rybicki, and Sasha Voitovych.
\newblock Near-optimal leader election in population protocols on graphs.
\newblock In {\em PODC}, page 246–256, 2022.

\bibitem{Angluin.10.1007/11502593_8}
Dana Angluin, James Aspnes, Melody Chan, Michael~J. Fischer, Hong Jiang, and Ren{\'e} Peralta.
\newblock Stably computable properties of network graphs.
\newblock In {\em DCOSS}, pages 63--74, 2005.

\bibitem{AngluinADFP2006}
Dana Angluin, James Aspnes, Zo{\"e} Diamadi, Michael~J. Fischer, and Ren\'e Peralta.
\newblock Computation in networks of passively mobile finite-state sensors.
\newblock {\em Distributed Computing}, 18(4):235--253, 2006.

\bibitem{AAE06}
Dana Angluin, James Aspnes, and David Eisenstat.
\newblock Stably computable predicates are semilinear.
\newblock In {\em PODC}, pages 292--299, 2006.

\bibitem{SimpleMajority}
Dana Angluin, James Aspnes, and David Eisenstat.
\newblock A simple population protocol for fast robust approximate majority.
\newblock In {\em DISC}, pages 20--32, 2007.

\bibitem{Epidemic}
Dana Angluin, James Aspnes, and David Eisenstat.
\newblock Fast computation by population protocols with a leader.
\newblock {\em Distributed Computing}, 21(3):183--199, 2008.

\bibitem{AAER07}
Dana Angluin, James Aspnes, David Eisenstat, and Eric Ruppert.
\newblock The computational power of population protocols.
\newblock {\em Distributed Computing}, 20(4):279--304, 2007.

\bibitem{Angluin.ss}
Dana Angluin, James Aspnes, Michael~J. Fischer, and Hong Jiang.
\newblock Self-stabilizing population protocols.
\newblock {\em ACM Trans. Auton. Adapt. Syst.}, 3(4), 2008.

\bibitem{AspnesBBS2016}
James Aspnes, Joffroy Beauquier, Janna Burman, and Devan Sohier.
\newblock Time and space optimal counting in population protocols.
\newblock In {\em OPODIS}, pages 13:1--13:17, 2017.

\bibitem{SSLEar2}
Joffroy Beauquier, Peva Blanchard, and Janna Burman.
\newblock Self-stabilizing leader election in population protocols over arbitrary communication graphs.
\newblock In {\em OPODIS}, pages 38--52, 2013.

\bibitem{count.space}
Joffroy Beauquier, Janna Burman, Simon Clavi{\`e}re, and Devan Sohier.
\newblock Space-optimal counting in population protocols.
\newblock In {\em DISC}, pages 631--646, 2015.

\bibitem{BenMajority}
Stav Ben-Nun, Tsvi Kopelowitz, Matan Kraus, and Ely Porat.
\newblock An o(log3/2 n) parallel time population protocol for majority with o(log n) states.
\newblock In {\em PODC}, page 191–199, 2020.

\bibitem{berenbrink_et_al:LIPIcs.DISC.2018.10}
Petra Berenbrink, Robert Els\"{a}sser, Tom Friedetzky, Dominik Kaaser, Peter Kling, and Tomasz Radzik.
\newblock {A Population Protocol for Exact Majority with O(log5/3 n) Stabilization Time and Theta(log n) States}.
\newblock In {\em DISC}, pages 10:1--10:18, 2018.

\bibitem{BEF+21}
Petra Berenbrink, Robert Els{\"a}sser, Tom Friedetzky, Dominik Kaaser, Peter Kling, and Tomasz Radzik.
\newblock Time-space trade-offs in population protocols for the majority problem.
\newblock {\em Distributed Computing}, 34(2):91--111, 2021.

\bibitem{BGK20}
Petra Berenbrink, George Giakkoupis, and Peter Kling.
\newblock Optimal time and space leader election in population protocols.
\newblock In {\em STOC}, page 119–129, 2020.

\bibitem{BilkeMajority}
Andreas Bilke, Colin Cooper, Robert Els\"{a}sser, and Tomasz Radzik.
\newblock Brief announcement: Population protocols for leader election and exact majority with o(log2 n) states and o(log2 n) convergence time.
\newblock In {\em PODC}, page 451–453, 2017.

\bibitem{CC19}
Hsueh-Ping Chen and Ho-Lin Chen.
\newblock Self-stabilizing leader election.
\newblock In {\em PODC}, pages 53--59, 2019.

\bibitem{CC20}
Hsueh-Ping Chen and Ho-Lin Chen.
\newblock Self-stabilizing leader election in regular graphs.
\newblock In {\em PODC}, pages 210--217, 2020.

\bibitem{DE19}
David Doty and Mahsa Eftekhari.
\newblock Efficient size estimation and impossibility of termination in uniform dense population protocols.
\newblock In {\em PODC}, page 34–42, 2019.

\bibitem{Majority.optimal}
David Doty, Mahsa Eftekhari, Leszek Gąsieniec, Eric Severson, Przemyslaw Uznański, and Grzegorz Stachowiak.
\newblock A time and space optimal stable population protocol solving exact majority.
\newblock In {\em FOCS}, pages 1044--1055, 2022.

\bibitem{DS18}
David Doty and David Soloveichik.
\newblock Stable leader election in population protocols requires linear time.
\newblock {\em Distributed Computing}, 31(4):257--271, 2018.

\bibitem{GSU19}
Leszek G\k{a}sieniec, Grzegorz Stachowiak, and Przemyslaw Uznanski.
\newblock Almost logarithmic-time space optimal leader election in population protocols.
\newblock In {\em SPAA}, page 93–102, 2019.

\bibitem{FastSpaceLE}
Leszek G\textalpha{}sieniec and Grzegorz Stachowiak.
\newblock Fast space optimal leader election in population protocols.
\newblock In {\em SODA}, page 265–266, 2018.

\bibitem{MertziosMajority}
George~B. Mertzios, Sotiris~E. Nikoletseas, Christoforos~L. Raptopoulos, and Paul~G. Spirakis.
\newblock Determining majority in networks with local interactions and very small local memory.
\newblock In {\em ICALP}, pages 871--882, 2014.

\bibitem{MajorityArbitrary}
George~B. Mertzios, Sotiris~E. Nikoletseas, Christoforos~L. Raptopoulos, and Paul~G. Spirakis.
\newblock Population protocols for majority in arbitrary networks.
\newblock In {\em Extended Abstracts Summer 2015}, pages 77--82, 2017.

\bibitem{MICHAIL20112434}
Othon Michail, Ioannis Chatzigiannakis, and Paul~G. Spirakis.
\newblock Mediated population protocols.
\newblock {\em Theoretical Computer Science}, 412(22):2434--2450, 2011.

\bibitem{networkconstruction}
Othon Michail and Paul~G. Spirakis.
\newblock Simple and efficient local codes for distributed stable network construction.
\newblock In {\em PODC}, page 76–85, 2014.

\bibitem{MICHAIL2022104698}
Othon Michail, Paul~G. Spirakis, and Michail Theofilatos.
\newblock Simple and fast approximate counting and leader election in populations.
\newblock {\em Information and Computation}, 285(A):104698, 2022.

\bibitem{MajorityCounting}
Yves Mocquard, Emmanuelle Anceaume, James Aspnes, Yann Busnel, and Bruno Sericola.
\newblock Counting with population protocols.
\newblock In {\em NCA}, pages 35--42, 2015.

\bibitem{MocMajority}
Yves Mocquard, Emmanuelle Anceaume, and Bruno Sericola.
\newblock Optimal proportion computation with population protocols.
\newblock In {\em NCA}, pages 216--223, 2016.

\bibitem{SudoLB}
Yuichi Sudo and Toshimitsu Masuzawa.
\newblock Leader election requires logarithmic time in population protocols.
\newblock {\em Parallel Processing Letters}, 30(01):2050005, 2020.

\bibitem{SUDO2012100}
Yuichi Sudo, Junya Nakamura, Yukiko Yamauchi, Fukuhito Ooshita, Hirotsugu Kakugawa, and Toshimitsu Masuzawa.
\newblock Loosely-stabilizing leader election in population protocol model.
\newblock In {\em SIROCCO}, pages 295--308, 2010.

\bibitem{sudotime}
Yuichi Sudo, Fukuhito Ooshita, Taisuke Izumi, Hirotsugu Kakugawa, and Toshimitsu Masuzawa.
\newblock Time-optimal leader election in population protocols.
\newblock {\em IEEE Transactions on Parallel and Distributed Systems}, 31(11):2620--2632, 2020.

\bibitem{LSLEonArbi}
Yuichi Sudo, Fukuhito Ooshita, Hirotsugu Kakugawa, and Toshimitsu Masuzawa.
\newblock Loosely stabilizing leader election on arbitrary graphs in population protocols without identifiers or random numbers.
\newblock {\em IEICE Transactions on Information and Systems}, E103.D(3):489--499, 2020.

\bibitem{SudoGlobal}
Yuichi Sudo, Masahiro Shibata, Junya Nakamura, Yonghwan Kim, and Toshimitsu Masuzawa.
\newblock Self-stabilizing population protocols with global knowledge.
\newblock {\em IEEE Transactions on Parallel and Distributed Systems}, 32(12):3011--3023, 2021.

\bibitem{yasumi_et_al:LIPIcs.OPODIS.2021.13}
Hiroto Yasumi, Fukuhito Ooshita, and Michiko Inoue.
\newblock Population protocols for graph class identification problems.
\newblock In {\em OPODIS}, pages 13:1--13:19, 2021.

\bibitem{BipartitionYasumi}
Hiroto Yasumi, Fukuhito Ooshita, Michiko Inoue, and Sébastien Tixeuil.
\newblock Uniform bipartition in the population protocol model with arbitrary graphs.
\newblock {\em Theoretical Computer Science}, 892:187--207, 2021.

\bibitem{YOKOTA.time}
Daisuke Yokota, Yuichi Sudo, and Toshimitsu Masuzawa.
\newblock Time-optimal self-stabilizing leader election on rings in population protocols.
\newblock {\em IEICE Transactions on Fundamentals of Electronics, Communications and Computer Sciences}, E104.A(12):1675--1684, 2021.

\bibitem{Yokota.near}
Daisuke Yokota, Yuichi Sudo, Fukuhito Ooshita, and Toshimitsu Masuzawa.
\newblock A near time-optimal population protocol for self-stabilizing leader election on rings with a poly-logarithmic number of states.
\newblock In {\em PODC}, page 2–12, 2023.

\end{thebibliography}

\clearpage
%%
%% If your work has an appendix, this is the place to put it.
%\appendix

\end{document}